\documentclass[letterpaper]{article} % DO NOT CHANGE THIS
\usepackage{aaai23}  % DO NOT CHANGE THIS
\usepackage{times}  % DO NOT CHANGE THIS
\usepackage{helvet}  % DO NOT CHANGE THIS
\usepackage{courier}  % DO NOT CHANGE THIS
\usepackage[hyphens]{url}  % DO NOT CHANGE THIS
\usepackage{graphicx} % DO NOT CHANGE THIS
\usepackage{natbib}  % DO NOT CHANGE THIS AND DO NOT ADD ANY OPTIONS TO IT
\usepackage{caption} % DO NOT CHANGE THIS AND DO NOT ADD ANY OPTIONS TO IT
\usepackage{algorithm}
\usepackage{algorithmic}

\usepackage{newfloat}
\usepackage{listings}
\DeclareCaptionStyle{ruled}{labelfont=normalfont,labelsep=colon,strut=off} % DO NOT CHANGE THIS
\floatstyle{ruled}
\newfloat{listing}{tb}{lst}{}
\floatname{listing}{Listing}
\newcommand{\stitle}[1]{\vspace{1ex} \noindent{\bf #1}}
\usepackage{subfigure}
\usepackage{multirow}
\usepackage{multicol}
\usepackage{booktabs} % for table \toprule、\midrule、\bottomrule
\usepackage{amssymb}
\usepackage{amsmath}

\newtheorem{theorem}{Theorem}
\newtheorem{proof}{Proof}
\newtheorem{definition}{Definition}

\newtheorem{lemma}{Lemma}

\newtheorem{example}{Example}

\title{Scalable and Effective Conductance-based Graph Clustering}
\author{
    Longlong Lin\textsuperscript{\rm 1}, Rong-Hua Li \textsuperscript{\rm 2}, Tao Jia\textsuperscript{\rm 1}\\
}
\affiliations{
    \textsuperscript{\rm 1}College of Computer and Information Science,
    Southwest University,
    Chongqing 400715, China \\
        \textsuperscript{\rm 2}Beijing Institute of Technology, China\\
longlonglin@swu.edu.cn; lironghuabit@126.com; tjia@swu.edu.cn
}

\usepackage{bibentry}
\begin{document}

\maketitle

\begin{abstract}
	Conductance-based graph clustering has been recognized as a fundamental operator in numerous graph analysis applications. Despite the significant success of conductance-based graph clustering, existing algorithms are either hard to obtain satisfactory clustering qualities, or have high time and space complexity to achieve provable clustering qualities. To overcome these limitations, we devise  a powerful \textit{peeling}-based graph clustering framework \textit{PCon}. We show that many existing solutions can be reduced to our framework. Namely, they first define a score function for each vertex, then iteratively remove the vertex with the smallest score.  Finally, they output the result with the smallest conductance during the peeling process. Based on our framework, we propose two novel algorithms \textit{PCon\_core} and \emph{PCon\_de} with linear time and space complexity, which can efficiently and effectively identify clusters from massive graphs with more than a few billion edges. Surprisingly, we prove that \emph{PCon\_de} can identify clusters with near-constant approximation ratio, resulting in an important theoretical improvement over the well-known quadratic Cheeger bound. Empirical results on real-life and synthetic datasets show that our algorithms can achieve 5$\sim$42 times speedup with a high clustering accuracy, while using 1.4$\sim$7.8 times less memory than the baseline algorithms.
	
\end{abstract} 

\section{Introduction}
Graph clustering is an important algorithmic primitive with applications in numerous tasks, including image segmentation \cite{DBLP:conf/cvpr/ShiM97, DBLP:conf/cvpr/TolliverM06}, community detection \cite{Fortunato2009Community,  DBLP:conf/www/LeskovecLM10},  and machine learning \cite{DBLP:conf/nips/BelkinN01, DBLP:conf/icml/BianchiGA20}.  Generally,  graph clustering aims to partition the entire graph into several non-overlapping vertex sets, called clusters, such that the vertices within the same cluster have more connections than the vertices in different clusters.  Many graph clustering methods have been proposed, such as modularity based graph clustering \cite{newman2004fast}, structural graph clustering \cite{DBLP:conf/kdd/XuYFS07}, and cohesive subgraph based clustering \cite{DBLP:conf/icde/ChangQ19}. Perhaps, the most representative graph clustering method is the \textit{conductance}-based clustering \cite{DBLP:conf/focs/AndersenCL06, DBLP:conf/sigmod/YangXWBZL19} due to its nice structure properties and solid theoretical foundation. %load balancing \cite{DBLP:books/siam/06/DevineBK06},

Given an undirected graph $G(V,E)$, the \textit{conductance} of the cluster $C$ is defined as $\phi(C)=\frac{|E(C,\bar{C})|}{\min\{vol(C), 2m-vol(C)\}}$, in which $|E(C,\bar{C})|$ is the number of edges with one endpoint in $C$ and another not in $C$, $vol(C)$ is the sum of the degree of all vertices in $C$, and $m$ is the number of the edges in $G$. Thus, a smaller $\phi(C)$ implies that  the number of edges going out of $C$ is relatively small compared with the number of edges within $C$. As a consequence, the smaller the $\phi(C)$, the better the clustering quality of the cluster $C$ \cite{DBLP:conf/www/LeskovecLM10, DBLP:journals/kais/YangL15}. Therefore,  \textit{conductance}-based graph clustering aims at identifying a vertex set $C$ with the smallest $\phi(C)$. However, identifying the smallest conductance $\phi^* \in (0,1]$ raises significant challenges due to its NP-hardness \cite{DBLP:journals/cc/ChawlaKKRS06}. On the positive side, many approximate or heuristic algorithms have been proposed to either reduce the conductance of the returned cluster or improve the efficiency. For example, classic \textit{Fiedler} vector-based spectral clustering algorithm outputs a cluster with conductance $O(\sqrt{\phi^*})$ \cite{alon1985lambda1}. Unfortunately, such an algorithm has to compute the eigenvector corresponding to the second smallest eigenvalue of normalized Laplacian matrix of $G$, resulting in prohibitively high time and space complexity.  To boost the efficiency, numerous diffusion-based local clustering algorithms have been proposed, whose running time depends
only on the size of the resulting cluster and is independent of or depends at most polylogarithmically on the size of the entire graph \cite{DBLP:conf/stoc/SpielmanT04, DBLP:conf/focs/AndersenCL06, DBLP:conf/kdd/KlosterG14}. However all these diffusion-based local clustering algorithms are heuristic  and the clustering quality of their  output is heavily dependent on many hard-to-tune parameters, resulting in that their clustering performance are unstable and in most cases very poor \cite{DBLP:conf/icml/ZhuLM13, DBLP:conf/icml/Fountoulakis0Y20}.  Thus, it is very  challenging to improve both  the computational efficiency and the clustering quality.

To this end, we propose a powerful \textit{peeling}-based computing framework \emph{PCon}, which can efficiently and effectively identify conductance-based clusters. In particular, we observe that \textit{Fiedler} vector-based spectral clustering algorithms and diffusion-based local clustering algorithms are  essentially a \textit{peeling}-based computing paradigm. Namely, they first define a score function for each vertex, then iteratively remove the vertex with the smallest score.  Finally, they output the results with the smallest conductance during the peeling process. Thus, the primary challenge is how to design a proper score function that is easier to compute and can obtain better clustering quality. With the help of the concept of degeneracy ordering \cite{DBLP:journals/corr/cs-DS-0310049}, we propose a heuristic algorithm \emph{PCon\_core} with linear time and space complexity. Specifically, \emph{PCon\_core} assigns the score  of each vertex as the vertex's degeneracy ordering. To further improve the clustering quality, we devise another effective algorithm \emph{PCon\_de}, which assigns the score of each vertex based on the \textit{degree ratio} (Definition \ref{def:dr}). \textit{Degree ratio} is a novel concept, which is the ratio of the degree of the vertex in the current subgraph to the degree of the vertex in the original graph. Surprisingly, we prove that \emph{PCon\_de} has a near-constant approximation ratio, which achieves an important theoretical improvement over the well-known quadratic chegger bound of \textit{Fiedler} vector-based spectral clustering. In a nutshell, our main contributions are listed as follows.
\begin{itemize}
\item \textbf{Novel Computing Framework:} We introduce \textit{PCon}, a
novel \textit{peeling}-based graph clustering framework, which can embrace most state-of-the-art algorithms.

\item \textbf{Novel algorithms:} Based on \textit{PCon}, we develop two  novel and efficient algorithms with linear time and space complexity. One is a heuristic algorithm which aims to optimize efficiency, while the other is an approximation algorithm which can obtain provable clustering qualities.

\item \textbf{Extensive Experiments:} We conduct extensive experiments on eleven datasets with six competitors to test the scalability and effectiveness of the proposed algorithms. The experimental results show that our proposal outperforms diffusion-based local clustering by a large margin in terms of clustering accuracy, and achieves 5$\sim$42 times speedup  with a high clustering  accuracy while using 1.4$\sim$7.8 times less memory than \textit{Fiedler} vector-based spectral clustering.  
\end{itemize}

\section{Problem Formulation} \label{sec:pro}

We use $G(V,E)$ to denote  an undirected graph, in which $V$ (resp. $E$) indicates the vertex set (resp. the edge set) of $G$.   Let $|V|=n$ and $|E|=m$ be the number of vertices and the number of edges, respectively. For simplicity, in this paper, we focus on un-weighted graphs\footnote{Because our results can be easily extended to the weighted case in a straightforward manner.}. Let $G_{S}=(S,E_S)$ be the subgraph induced by  $S$ if $S \subseteq V$ and $E_S= \{(u,v)\in E| u,v \in S\}$. We use $N_{S}(v)=\{u\in S|(u,v)\in E\}$ to denote the neighbors of $v$  in $S$. Let $d_S(v)=|N_S(v)|$ be the degree of vertex $v$ in $S$. When the context is clear, we use $N(v)$ and $d(v)$ to represent $N_V(v)$ and $d_V(v)$, respectively. For a vertex subset $S\subseteq V$,  we define $\bar{S}=V\setminus S$ as the complement of $S$. We use $E(S, \bar{S})=\{(u,v) \in E|u\in S, v\in \bar{S}\}$ to represent the edges with one endpoint in $S$ and another not in $S$. Let $vol(S)=\sum_{u \in S}d(u)$ be the sum of the degree of all vertices in $S$.

A cluster is a vertex set $S \subseteq V$. According to \cite{DBLP:conf/www/LeskovecLM10, DBLP:journals/kais/YangL15}, we know that a cluster $S$ is good if the cluster is densely connected internally and well separated from the remainder of $G$. Therefore, we use a representative metric \emph{conductance} \cite{fiedler1973algebraic, DBLP:conf/stoc/SpielmanT04, DBLP:conf/focs/AndersenCL06, DBLP:conf/kdd/KlosterG14} to measure the quality of a cluster $S$.

\begin{definition} \label{def:co}
Given an undirected graph $G(V,E)$ and a vertex subset $S \subsetneq V$, the conductance of $S$ is defined as	$\phi(S)=\frac{|E(S,\bar{S})|}{\min\{vol(S), 2m-vol(S)\}}$.  $\phi(V)=1$ for convenience.
\end{definition}

By Definition \ref{def:co}, we have $\phi(S)=\phi(\bar{S})$. Figure \ref{fig:intro} illustrates the concept of \textit{conductance} on a synthetic graph.

\stitle{Problem Statement.} Given an undirected graph $G(V,E)$, the goal of the \textit{conductance}-base graph clustering is to identify a vertex subset $S^* \subseteq V$, satisfying $vol(S^{*})\leq m$ and  $\phi(S^*)\leq \phi(S)$ for any $S \subseteq V$. For simplicity, we use $\phi^*$ to represent  $\phi(S^*)$.

\begin{figure}[t]
	\centering
	\includegraphics[width=0.8\columnwidth]{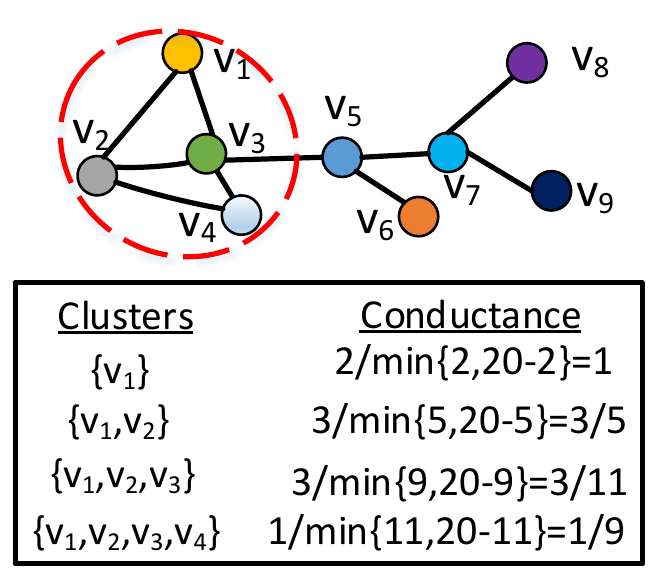} 
	\caption{Illustration of conductance for a synthetic graph with 9 vertices and 10 edges.  The conductance of several clusters is shown in the box. The smallest conductance $\phi^*=1/9$ and the corresponding cluster is $\{v_1, v_2,v_3,v_4\}$ included in the red circle.} \vspace{-0.3cm}
	\label{fig:intro} 
\end{figure}

\section{Existing Solutions} \label{sec:existing}
 Here, we review some state-of-the-art algorithms. For convenience, we  classify these algorithms into three categories: \textit{Fiedler} vector-based spectral clustering, diffusion-based local clustering, and other methods.

\begin{table*}[t!]
	\centering
	\caption{A comparison of state-of-the-art algorithms. $\phi^*$ is the smallest conductance value and $\phi^*\in (0,1]$. $\epsilon$ is the error tolerance. $\alpha$ and $t$ are model parameters of $\textit{NIBBLE\_PPR}$ and \textit{HK\_Relax}, respectively. The space complexity of SC is $O(n^2)$ and the rest is $O(m+n)$. $\times$ represents the corresponding method has no accuracy guarantee.}
	\scalebox{1}{
		\begin{tabular}{c|c|c|c}
			\toprule
			\multicolumn{1}{c|}{Methods} & \multicolumn{1}{c|}{Accuracy Guarantee}&
			\multicolumn{1}{c|}{Time Complexity}& \multicolumn{1}{c}{Remark}\\
			\midrule
			\textit{SC} \cite{fiedler1973algebraic} & $O(\sqrt{\phi^*})$ & $O(n^3)$ & Eigenvector-based \\
			\textit{ASC} \cite{trevisan2017lecture} & $O(\sqrt{4\phi^*+2\epsilon})$ & $O((m+n) \frac{1}{\epsilon} \log \frac{n}{\epsilon})$& Eigenvector-based\\
			\midrule
			\textit{NIBBLE} \cite{DBLP:conf/stoc/SpielmanT04} & $\times$ &$O(\frac{\log \frac{1}{\epsilon}}{\epsilon})$& Diffusion-based\\
			\textit{NIBBLE\_PPR} \cite{DBLP:conf/focs/AndersenCL06} & $\times$ & $O(\frac{\log \frac{1}{\alpha\epsilon}}{\alpha \epsilon})$& Diffusion-based\\
			\textit{HK\_Relax} \cite{DBLP:conf/kdd/KlosterG14} & $\times$ & $O(\frac{te^t \log(\frac{1}{\epsilon})}{\epsilon})$& Diffusion-based\\
			\midrule
			\textit{PCon\_core} (this paper)& $\times$ & $O(m+n)$& Degeneracy-based\\
						\textit{PCon\_de} (this paper) & $O(1/2+1/2\phi^*)$ & $O(m+n)$& Degree Ratio-based\\
			\bottomrule	
	\end{tabular}}
	\label{tab:alg}
\end{table*}

\subsection{\textit{Fiedler} Vector-based Spectral Clustering} \label{subsec:com_tppr}

\textit{Fiedler} vector-based spectral clustering obtains the clustering result by calculating the spectrum of the normalized Laplacian matrix. Specifically, let $A$ be the adjacency matrix of $G$, where  $A_{uv}=1$ if $(u,v) \in E$, and $A_{uv}=0$ otherwise. We use $D$ to represent the diagonal degree matrix of $G$, in which $D_{uu}=d(u)$. The Laplacian matrix of $G$ is defined as $L=D-A$. Furthermore, the normalized Laplacian matrix is defined as $\mathcal{L}=D^{-1/2}LD^{-1/2}$. The following theorem is an important theoretical basis of \textit{Fiedler} vector-based spectral clustering.

\begin{theorem} [Cheeger inequality \cite{alon1985lambda1}]
Given an undirected graph $G(V,E)$ and its normalized Laplacian matrix $\mathcal{L}$, we assume that $\lambda_2$ is the second smallest eigenvalue of $\mathcal{L}$. Then, we have $\lambda_2/2\leq \phi^*\leq \sqrt{2\lambda_2}$.
\end{theorem}

At a high level, \textit{Fiedler} vector-based spectral clustering consists of three steps: (1) Compute the eigenvector $x$ of $\lambda_2$. (2) Sort all entries in $x$ such that $x_1\leq x_2 \leq ... \leq x_n$. (3) Output $S=\arg\min \phi(S_i)$, in which $S_i=\{x_1,x_2,...,x_i\}$. 
%The detailed procedure of the \textit{Fiedler} vector-based spectral clustering is shown in Algorithm \ref{alg:sc}. Note that in Lines 6-10, we output the smaller of $S$ and $\bar{S}$ due to $\phi(S)=\phi(\bar{S})$.

\begin{theorem} [\cite{fiedler1973algebraic, alon1985lambda1}]
Let  $S$ be the vertex subset returned by \textit{Fiedler} vector-based spectral clustering, we have $\phi(S)\leq \sqrt{2\lambda_2}$. Furthermore, we have $\phi^*\leq \phi(S) \leq 2\sqrt{\phi^*}$.
\end{theorem}

\iffalse
\begin{algorithm}[tb]
	\caption{\textit{Fiedler} vector-based spectral clustering}
	\label{alg:sc}
	\textbf{Input}: An undirected graph $G(V, E)$ \\
	\textbf{Output}: A cluster $S$
	\begin{algorithmic}[1] %[1] enables line numbers
		\STATE $L \leftarrow D-A$
		\STATE $\mathcal{L} \leftarrow D^{-1/2}LD^{-1/2}$
		\STATE $x \leftarrow $ eigenvector of second smallest eigenvalue of $\mathcal{L}$ 
		\STATE $x_i \leftarrow $ to be index of $x$ with $i$th smallest value
		\STATE $S \leftarrow \arg\min \phi(S_i)$, where $S_i=\{x_1,x_2,...,x_i\}$
		\IF {$|S|>|\bar{S}|$}
		\STATE \textbf{return} $\bar{S}$
		\ELSE
		\STATE \textbf{return} $S$
		\ENDIF
	\end{algorithmic}
\end{algorithm}
\fi

Note that since \textit{Fiedler} vector-based spectral clustering (\textit{SC} for short) needs to calculate the eigenvector of $\mathcal{L}$, its time complexity is $O(n^3)$ and space complexity is $O(n^2)$ \cite{fiedler1973algebraic}, resulting in poor scalability. Thus, some authors devised efficient approximate algorithms. For example, Trevisan proposed \textit{ASC} to approximate the  eigenvector of second smallest eigenvalue of $\mathcal{L}$  \cite{trevisan2017lecture}. As a result, \textit{ASC} identifies a cluster $S$ with $\phi(S)\leq \sqrt{4\phi^*+2\epsilon}$ using $O((m+n)\cdot \frac{1}{\epsilon} \log \frac{n}{\epsilon})$ time cost, in which $\epsilon$ is the error tolerance. Although the time complexity of \textit{ASC} is much  less than of \textit{SC}, the quality of \textit{ASC} decreases rapidly as shown in our experiments.

\subsection{Diffusion-based Local Clustering} \label{subsec:gr}
Diffusion-based local clustering is another  technique that propagates information from the given query seed vertex $q$ to identify clusters. Here, we state three state-of-the-art graph diffusions: Truncated Random Walk \cite{DBLP:conf/stoc/SpielmanT04}, Personalized PageRank \cite{ DBLP:conf/focs/AndersenCL06}, and Heat Kernel \cite{DBLP:conf/kdd/KlosterG14}. For simplicity, we use $\pi_T$, $\pi_P$, and $\pi_H$ to represent the probability distribution at the end of Truncated Random Walk, Personalized PageRank, and Heat Kernel, respectively. 

Before proceeding further, we give some important notations. Let $P=D^{-1}A$ be the probability transition matrix of $G$, in which $P_{uv}=1/d(u)$ for any $v \in N(u)$. Moreover, we use $P^{k}$ to represent the $k$-hop probability  transition matrix of $G$. Namely, $P^{k}_{uv}$ is the probability that a $k$-hop ($k\geq 1$) random walk from vertex $u$ would end at vertex $v$.

1) Truncated Random Walk \cite{DBLP:conf/stoc/SpielmanT04} is a graph diffusion algorithm where propagation and truncation are alternately performed. Specifically, for any vector $s$ and error tolerance $\epsilon$, we define a truncation operator $Tr(s)$ on $u$ as follows:

\begin{equation}
Tr(s)[u]=
\begin{cases}
s[u], &if \quad s[u]\geq d(u)\epsilon\\
0, &Otherwise\\
\end{cases}
\end{equation}

Furthermore, we use $Z_0=\chi_q$ to represent the one-hot vector with only a value-1 entry corresponding to the given query seed $q$. The propagation and truncation are denoted as $Z_{i}=Tr(Z_{i-1}P)$, in which $Z_{i-1}P$ is the propagation process and  $Tr(Z_{i-1}P)$ is the truncation process which can obtain the next probability distribution. Thus, we let $\pi_T=Z_{N}$ be the probability distribution after $N$ iterations, in which $N$ is an input parameter.

2) Personalized PageRank \cite{ DBLP:conf/focs/AndersenCL06} models a special random walk process. Specifically, given a stop probability parameter $\alpha$ (a.k.a. teleportation probability), we denote a $\alpha$-discount random walk as follows: (1)
It starts from the given query seed $q$. (2) At each step it stops in the current vertex with
probability $\alpha$, or it continues to walk according to the probability transition matrix $P$ with 1-$\alpha$. Thus, $\pi_P(u)$ is the probability that the $\alpha$-discount random walk stops in $u$. %Namely, $\pi_P(u)=\sum\limits_{k=0}^{\infty}\alpha(1-\alpha)^k P^{k}_{qu}$.

3) Heat Kernel \cite{DBLP:conf/kdd/KlosterG14} also models a special random walk process. Specifically, given a heat constant $t$, $\pi_H(u)$ is the probability that a random walk of length $k$ starting from the given query seed $q$ would end at the vertex $u$, in which $k$ is sampled from the Poisson distribution $\eta(k)=\frac{e^{-t}t^k}{k!}$. Thus  $\pi_H(u) =\sum\limits_{k=0}^{\infty}\eta(k) P^{k}_{qu}$.

Similarly, these diffusion-based local clustering algorithms also consist of three steps: (1) Compute the probability distribution $\pi$ at the end of the corresponding graph diffusion, and $y=\pi D^{-1}$. (2) Sort all non-zero entries in $y$ such that $y_1\geq y_2 \geq ... \geq y_{sup(y)}$, in which $sup(y)$ is the number of the non-zero entries in $y$. (3) Output $S=\arg\min \phi(S_i)$, in which $S_i=\{y_1,y_2,...y_i\}$. 
%The detailed procedure of the local clustering is shown in Algorithm \ref{alg:local}.

Note that since diffusion-based local clustering algorithms aim to recover the cluster to which the given query seed $q$ belongs, they only have locally-biased Cheeger-like quality \cite{DBLP:conf/stoc/SpielmanT04, DBLP:conf/focs/AndersenCL06, DBLP:conf/kdd/KlosterG14}. Namely, diffusion-based local clustering algorithms do not give the theoretical gap to $\phi^*$. Besides, the clustering quality of their output is heavily dependent on the given query seed and hard-to-tune parameters, resulting in that they are unstable and prone to finding degenerate solutions in most cases \cite{DBLP:conf/icml/ZhuLM13, DBLP:conf/icml/Fountoulakis0Y20}.

\iffalse
\begin{algorithm}[tb]
	\caption{Local Clustering}
	\label{alg:local}
	\textbf{Input}: An undirected  graph $G(V, E)$, a query seed $q$, and parameter domain $\Theta$ \\
	\textbf{Output}: A cluster $S$
	\begin{algorithmic}[1] %[1] enables line numbers
		\STATE $\pi \leftarrow$ the probability distribution after the end of the corresponding graph diffusion \qquad  // $\pi$ is $\pi_T$ or $\pi_P$ or $\pi_H$
		\STATE $y \leftarrow \pi D^{-1}$  
		\STATE $y_i \leftarrow $ to be index of $y$ with $i$th largest value
		\STATE $S \leftarrow \arg\min \phi(S_i)$, where $S_i=\{y_1,y_2,...y_{sup(y)}\}$
		\IF {$|S|>|\bar{S}|$}
		\STATE \textbf{return} $\bar{S}$
		\ELSE
		\STATE \textbf{return} $S$
		\ENDIF
	\end{algorithmic}
\end{algorithm}
\fi

\subsection{Other Methods} 
Optimization-based \cite{leighton1999multicommodity, DBLP:conf/stoc/AroraRV04}  and flow-based \cite{DBLP:conf/soda/OrecchiaZ14, veldt2016simple, DBLP:conf/icml/WangFHMR17} are two important techniques for improving \textit{Fiedler} vector-based spectral clustering and diffusion-based local clustering. For example,  Leighton et al. \cite{leighton1999multicommodity} used linear programming to solve the conductance-based graph clustering with $O(\log n)$-approximation. Furthermore, Arora et al. \cite{DBLP:conf/stoc/AroraRV04} achieved the best $O(\sqrt{\log n})$-approximation using a non-trivially semi-definite programming algorithm with $O(\frac{(n+m)^{2}}{\epsilon^{O(1)}})$ time cost. However, all these  methods are mostly of theoretical interests only, as they are difficult to implement and offer rather poor practical efficiency \cite{DBLP:conf/sigmod/YangXWBZL19}.

\section{The Proposed Algorithms} \label{sec:our}
Here, we first present a three-stage computing framework, which embraces the computing paradigm of most existing approaches. Then, we develop two scalable and effective algorithms to solve the \textit{conductance}-based graph clustering.

\subsection{Computing Framework} 
Inspired by most existing approaches \cite{fiedler1973algebraic,DBLP:conf/stoc/SpielmanT04, DBLP:conf/focs/AndersenCL06, DBLP:conf/kdd/KlosterG14}, we  propose a three-stage computing framework \textit{PCon}. In \textit{Stage 1}, we  give a pre-defined  score function for every vertex according to the corresponding applications. For simplicity, we use $s(u)$ to represent the score of vertex $u$. In \textit{Stage 2}, we iteratively remove the vertex with the smallest score. Such an iterative deletion process is referred to as a peeling process.   In \textit{Stage 3}, we output the result with smallest conductance during the \emph{peeling} process.  

Obviously, \textit{Stage 1} is key to our proposed  computing framework. Different algorithms have different score functions. For example, $s(u)=x[u]$ for \textit{Fiedler} vector-based spectral clustering, in which $x$ is the eigenvector of second smallest eigenvalue of $\mathcal{L}$. In diffusion-based local clustering, we can know that $s(u)=y[u]$, in which $y=\pi D^{-1}$ and $\pi$ is the probability distribution of the diffusion process. Thus, most state-of-the-art algorithms can be reduced to the proposed three-stage computing framework. However, the primary problem with these existing algorithms to partition the graph is that it is difficult to efficiently and effectively obtain the score function. To this end, in the next sections, we devise two new score functions, which are easier  to compute.

\begin{algorithm}[tb]
	\caption{\textit{PCon\_core}}
	\label{alg:core}
	\textbf{Input}: An undirected  graph $G(V, E)$\\
	\textbf{Output}: A cluster $S$
	\begin{algorithmic}[1] %[1] enables line numbers
		\STATE $\{u_1,u_2,...,u_n\} \leftarrow$ the degeneracy ordering 
		\STATE $i\leftarrow n$, $S_i \leftarrow V$, $S \leftarrow V$
		\WHILE {$i \neq 0$}
		\IF {$vol(S_i)\leq m$ and $\phi(S_i)<\phi(S)$}
		\STATE $S \leftarrow S_i$
		\ENDIF
		\STATE  $S_{i-1} \leftarrow S_{i} \setminus \{u_{i}\}$, $i\leftarrow i-1$
		\ENDWHILE
		\STATE \textbf{return} $S$
	\end{algorithmic}
\end{algorithm}

\subsection{The \textit{PCon\_core} Algorithm}
Recall that in \textit{Stage 2}, we need to iteratively remove the vertex with the smallest score. Namely, we have to create a linear ordering on vertices.  Many ordering strategies have been proposed for numerous graph analysis tasks, such as graph coloring \cite{DBLP:conf/spaa/HasenplaughKSL14} and  $k$-clique listing \cite{DBLP:journals/pvldb/LiGQWYY20}. However, ordering techniques for conductance-based graph clustering are less explored. To fill this gap, we propose a \textit{simple} algorithm with the help of the well-known degeneracy ordering \cite{DBLP:journals/corr/cs-DS-0310049}.

\begin{definition} [Degeneracy ordering] \label{def:do}
	Given an undirected  graph $G$, a permutation $(u_1,u_2,...,u_n)$ on all vertices of $G$ is a degeneracy ordering iff every vertex $u_i$ has the minimum degree in the subgraph of $G$ induced by $\{u_i,u_{i+1},...,u_n\}$, that is, $u_i=\arg\min \{d_{V_i}(u)|u\in V_i\}$ where $V_i=\{u_i,u_{i+1},...,u_n\}$.
\end{definition}

We use degeneracy ordering to assign the score $s(u)$ of vertex $u$. Specifically, $u$ is the $s(u)$-th element in the degeneracy ordering from left to right.

\begin{figure}[t]
	\centering
	\includegraphics[width=0.9\columnwidth]{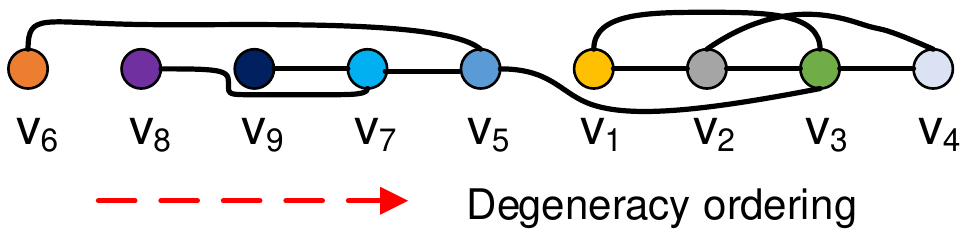} 
	\caption{A degeneracy ordering of vertices in Figure \ref{fig:intro}.}
	\label{fig:example}
	\vspace{-0.3cm}
\end{figure} 

\begin{example}
Reconsider the graph in Figure \ref{fig:intro}. According to Definition \ref{def:do}, we can derive the degeneracy ordering as $\{v_6,v_8,v_9,v_7,v_5,v_1,v_2,v_3,v_4\}$ which is illustrated  in	Figure \ref{fig:example}. Consider a vertex $v_1$, we can know that $s(v_1)=6$. This is because that $v_1$ has the minimum degree in the subgraph of $G$ induced by $\{v_1,v_2,v_3,v_4\}$. Similarity, we have $s(v_6)=1$ and $s(v_7)=4$.
\end{example}

The degeneracy ordering can be efficiently computed within linear time by the classic core-decomposition \cite{DBLP:journals/corr/cs-DS-0310049}. Specifically, it iteratively removes the vertex with the smallest degree in the current subgraph until all vertices are removed. When a vertex is removed, the degree of other vertices is updated accordingly. As a consequence, the sequence of the removed vertices forms  the degeneracy ordering. Using a proper data structure (e.g., bin-sort), we can implement the above vertex-removal process in linear time  \cite{DBLP:journals/corr/cs-DS-0310049}.  Based on the degeneracy ordering, we propose a \textit{simple but practically effective} algorithm \textit{PCon\_core}, which is outlined in Algorithm \ref{alg:core}. Specifically, we first set the score  $s(u)$ for each vertex $u\in V$ according to the degeneracy ordering (Line 1). Then, we execute \textit{Stage 2} (Lines 3, 7, and 8) and \textit{Stage 3} (Lines 4-6 and 9).

\subsection{The \textit{PCon\_de} Algorithm}
As stated in the Introduction, \textit{Fiedler} vector-based spectral clustering can obtain a cluster with conductance of $O(\sqrt{\phi^*})$. However, \textit{Fiedler} vector-based spectral clustering has prohibitively high time and space complexity. On the other hand, although diffusion-based local clustering has a very low time complexity, it is heuristic without a global conductance guarantee. Thus, an interesting problem is to devise an algorithm that has  better time complexity than \textit{Fiedler} vector-based spectral clustering and better conductance quality than diffusion-based local clustering.  To this end, we propose a novel algorithm \textit{PCon\_de} with linear time and space complexity, which has a near-constant approximation ratio.

\begin{lemma} \label{lem:1}
Given an undirected graph $G(V,E)$ and a vertex subset $S$, we have $\phi(S)=1-\frac{\sum\limits_{u\in S}d_{S}(u)}{\sum\limits_{u\in S}d_V(u)}$ if $vol(S)\leq m$.
\end{lemma}
\begin{proof}
According to Definition \ref{def:co}, if $vol(S)\leq m$, we have $\phi(S)=\frac{|E(S,\bar{S})|}{\min\{vol(S), 2m-vol(S)\}}=\frac{|E(S,\bar{S})|}{vol(S)}$. Furthermore, $\frac{|E(S,\bar{S})|}{vol(S)}=\frac{\sum\limits_{u\in S}d_{\bar{S}}(u)}{\sum\limits_{u\in S}d_{V}(u)}=\frac{\sum\limits_{u\in S}d_{\bar{S}}(u)+\sum\limits_{u\in S}d_{V}(u)-\sum\limits_{u\in S}d_{V}(u)}{\sum\limits_{u\in S}d_{V}(u)}=\frac{\sum\limits_{u\in S}d_{V}(u)-\sum\limits_{u\in S}d_{S}(u)}{\sum\limits_{u\in S}d_{V}(u)}$. Thus, $\phi(S)=1-\frac{\sum\limits_{u\in S}d_{S}(u)}{\sum\limits_{u\in S}d_V(u)}$.
\end{proof}

For simplicity, we denote a function $g(S)=\frac{\sum\limits_{u\in S}d_{S}(u)}{2\sum\limits_{u\in S}d_V(u)}$ and assume that the larger the value of $g(S)$, the better the quality of $S$. Generally, we let $\widetilde{S}$ be the optimal vertex set for $g(.)$. That is, $g(\widetilde{S})\geq g(S)$ for any vertex subset $S\subseteq V$.

\begin{lemma} \label{lem:2}
Given an undirected graph $G(V,E)$, we have  $\frac{d_{\widetilde{S}}(u)}{d_V(u)}\geq g(\widetilde{S})$ for any $u \in \widetilde{S}$ holds.
\end{lemma}
\begin{proof}
  This lemma can be proved by contradiction. Assume that there is a vertex $u \in \widetilde{S}$ such that $\frac{d_{\widetilde{S}}(u)}{d_V(u)}< g(\widetilde{S})$, we have $d_{\widetilde{S}}(u)< g(\widetilde{S}) d_V(u)$. Thus, $g(\widetilde{S} \setminus \{u\})=\frac{1/2\sum\limits_{v\in \widetilde{S}}d_{\widetilde{S}}(v)-d_{\widetilde{S}}(u)}{\sum\limits_{v\in \widetilde{S}}d_V(v)-d_V(u)}>\frac{1/2\sum\limits_{v\in \widetilde{S}}d_{\widetilde{S}}(v)-g(\widetilde{S}) d_V(u)}{\sum\limits_{v\in \widetilde{S}}d_V(v)-d_V(u)}=\frac{g(\widetilde{S})\sum\limits_{v \in \widetilde{S}}d_V(v)-g(\widetilde{S}) d_V(u)}{\sum\limits_{v\in \widetilde{S}}d_V(v)-d_V(u)}=g(\widetilde{S})$, which contradicts that $\widetilde{S}$ is the optimal vertex set for $g(.)$. As a result, $\frac{d_{\widetilde{S}}(u)}{d_V(u)}\geq g(\widetilde{S})$ for any $u \in \widetilde{S}$ holds.
\end{proof}

\begin{definition} [Degree ratio] \label{def:dr}
Given an undirected graph $G(V,E)$ and a subgraph $G_S$, the degree ratio of $u$ $\in S$ w.r.t. $G$ and $G_S$ is defined as $Dr_S(u)=\frac{d_S(u)}{d_V(u)}$.
\end{definition}

\begin{algorithm}[t]
	\caption{\textit{PCon\_de}} \label{algor:de}
	\textbf{Input:}
	An undirected graph $G(V,E)$\\
	\textbf{Output:}
	A cluster $\hat{S}$
	\begin{algorithmic}[1]
		\STATE $S\leftarrow V$;	$\hat{S}\leftarrow V$
		\STATE 	$Dr_S(u)\leftarrow \frac{d_S(u)}{d_V(u)}$ for each vertex $u \in S$
		
		\WHILE{$S\neq \emptyset$}
		\STATE $u \leftarrow \arg \min\{Dr_S(u)|u\in S\}$
		\STATE $S\leftarrow S \setminus \{u\}$
		\IF{$g(S)> g(\hat{S})$ and $vol(S)\leq m$}
		\STATE $\hat{S} \leftarrow S$
		\ENDIF
		\FOR {$v \in N_S(u)$}
		\STATE $Dr_S(v) \leftarrow Dr_S(v)- \frac{1}{d_V(v)}$
		\ENDFOR
		\ENDWHILE
		\STATE \textbf{return}  $\hat{S}$
	\end{algorithmic}
\end{algorithm}

Based on the above lemmas and definitions, we  devise a linear time greedy removing  algorithm called \textit{PCon\_de}, which is shown in Algorithm \ref{algor:de}. In particular,  Algorithm \ref{algor:de} first initializes  the current search space $S$ as $V$, candidate result $\hat{S}$ as $V$, and the \textit{degree ratio} $Dr_S(u)$ for every vertex $u\in S$ according to Definition \ref{def:dr} (Lines 1-2). Subsequently, it executes the greedy removing process in each round to improve the quality of the target cluster (Lines 3-12). Specifically, in each round, it obtains one vertex $u$ with the smallest \textit{degree ratio} (Line 4). Lines 5-11 update the current search space $S$, the candidate result $\hat{S}$ if any, and $Dr_S(v)$ for $v\in N_S(u)$. The iteration terminates once the current search space is empty (Line 3). Finally, it returns  $\hat{S}$ as the approximate solution (Line 13).

\begin{theorem}
Algorithm \ref{algor:de} can identify a cluster with conductance $1/2+1/2\phi^*$. 
\end{theorem}
\begin{proof}
	Let $\widetilde{S}$ is the optimal vertex set for $g(.)$. In Lines 3-12, Algorithm \ref{algor:de} executes the  peeling process. That is, in each round, it greedily deletes the vertex with the smallest \textit{degree ratio}. Consider the round $t$ when the first vertex $u$ of $\widetilde{S}$ is deleted. Let $V_t$ be the vertex set from the beginning of round $t$. Clearly, $\widetilde{S}$ is the subset of $V_t$ because $u$ is the first deleted vertex of $\widetilde{S}$. This implies that
	$\min\limits_{v \in V_t}Dr_{V_t}(v)=Dr_{V_t}(u)=\frac{d_{V_t}(u)}{d_V(u)}\geq \frac{d_{\widetilde{S}}(u)}{d_V(u)}$. Furthermore, $g(V_t)=\frac{\sum\limits_{v\in V_t}d_{V_t}(v)}{2\sum\limits_{v\in V_t}d_V(v)}=\frac{\sum\limits_{v\in V_t}d_V(v)\frac{d_{V_t}(v)}{d_V(v)}}{2\sum\limits_{v\in V_t}d_V(v)}\geq \frac{\sum\limits_{v\in V_t}d_V(v)\frac{d_{V_t}(u)}{d_V(u)}}{2\sum\limits_{v\in V_t}d_V(v)}= \frac{d_{V_t}(u)}{2d_V(u)}\geq \frac{d_{\widetilde{S}}(u)}{2d_V(u)}$. By Lemma \ref{lem:2}, we have $\frac{d_{\widetilde{S}}(u)}{d_V(u)}\geq g(\widetilde{S})$. Thus, $g(V_t)\geq \frac{g(\widetilde{S})}{2}$. Since  Algorithm \ref{algor:de} maintains the optimal solution during the peeling  process in Lines 6-8,  $\hat{S}$ will be returned in Line 13 and $g(\hat{S})\geq g(V_t)\geq  \frac{g(\widetilde{S})}{2}$. On the other hand, by Lemma \ref{lem:1}, we have $g(\hat{S})=\frac{1-\phi(\hat{S})}{2}$ and $g(S^*)=\frac{1-\phi(S^*)}{2}$. According to the definition of $\widetilde{S}$, we know that $g(\widetilde{S})\geq g(S^*)$. Thus, $\frac{1-\phi(\hat{S})}{2}=g(\hat{S})\geq \frac{g(\widetilde{S})}{2}\geq \frac{g(S^*)}{2}=\frac{1-\phi(S^*)}{4}$. Namely, $\phi(\hat{S})\leq 1-\frac{1-\phi(S^*)}{2}=1/2+1/2\phi^*$. As a result, Algorithm \ref{algor:de} can identify a cluster with conductance  $1/2+1/2\phi^*$.
\end{proof}

\section{Empirical Results}\label{sec:experiments}
\subsection{Experimental Setup}
We evaluate our  proposed solutions on six real-life publicly-available datasets\footnote{All datasets can be downloaded from  http://snap.stanford.edu/} (Table \ref{tab:data}), which are widely used benchmarks for conductance-based graph clustering \cite{DBLP:journals/pvldb/ShunRFM16, DBLP:conf/sigmod/YangXWBZL19}. The maximum connected components of these datasets are used in the experiments. We also use five synthetic graphs \textit{LFR} \cite{lancichinetti2009detecting}, \textit{WS} \cite{watts1998collective}, \textit{PLC} \cite{holme2002growing}, \textit{ER} \cite{erdos1960evolution}, and \textit{BA} \cite{barabasi1999emergence}.  The following six competitors are implemented for comparison.

\begin{table}[t!]
	\centering
	\caption{Dataset statistics. $\bar{d}$ is the average degree. } 
	\scalebox{1}{
		\begin{tabular}{c|ccccc}
			\toprule
			Dataset & $|V|$  & $|E|$  &  $\bar{d}$ \\
			\midrule
			DBLP & 317,080  & 1,049,866 &6.62\\
			Youtube & 1,134,890  & 2,987,624 &5.27 \\
			Pokec & 1,632,803  & 22,301,964 &2.73\\
			LJ & 4,843,953  & 42,845,684 &17.69 \\   
			Orkut & 3,072,441  & 117,185,083 &76.28\\
			Twitter & 41,652,231  & 1,202,513,046 &57.74&\\
			\bottomrule		
	\end{tabular}}
	\label{tab:data}
\end{table}

\begin{itemize}
\item  Eigenvector-based methods: \textit{SC} \cite{fiedler1973algebraic} and \textit{ASC} \cite{trevisan2017lecture}.  \textit{SC} (resp. \textit{ASC}) used  the exact eigenvector (resp. approximate eigenvector)  of the second smallest eigenvalue of $\mathcal{L}$ to execute the \textit{peeling} process. We use sparse matrices to  highly optimize these algorithms. 

\item Diffusion-based methods: 
\textit{NIBBLE\_PPR} \cite{DBLP:conf/focs/AndersenCL06} and
\textit{HK\_Relax} \cite{DBLP:conf/kdd/KlosterG14}. Since \textit{NIBBLE\_PPR} and \textit{HK\_Relax} took a seed vertex as input, to be more reliable,  we randomly select 50 vertices as seed vertices and report the average runtime and quality.  Unless specified otherwise, following previous work \cite{DBLP:journals/pvldb/ShunRFM16}, we set $\alpha=0.01$
and $\epsilon=\frac{1}{m}$ for  \textit{NIBBLE\_PPR}; $t=10$ and $\epsilon=\frac{1}{m}$ for  \textit{HK\_Relax}. Note that we do not include \textit{NIBBLE} of Table \ref{tab:alg} in the
experiments because \textit{NIBBLE} is outperformed by \textit{NIBBLE\_PPR} and \textit{HK\_Relax}  \cite{DBLP:journals/pvldb/ShunRFM16}. 

\item Flow-based methods: \textit{SimpleLocal} \cite{veldt2016simple} and \textit{CRD} \cite{DBLP:conf/icml/WangFHMR17}. These methods devise specialized max-flow algorithms with early termination. We take their default parameters in our experiments.
\end{itemize}

\begin{figure}[t!]
	\centering
	\includegraphics[width=0.45\textwidth]{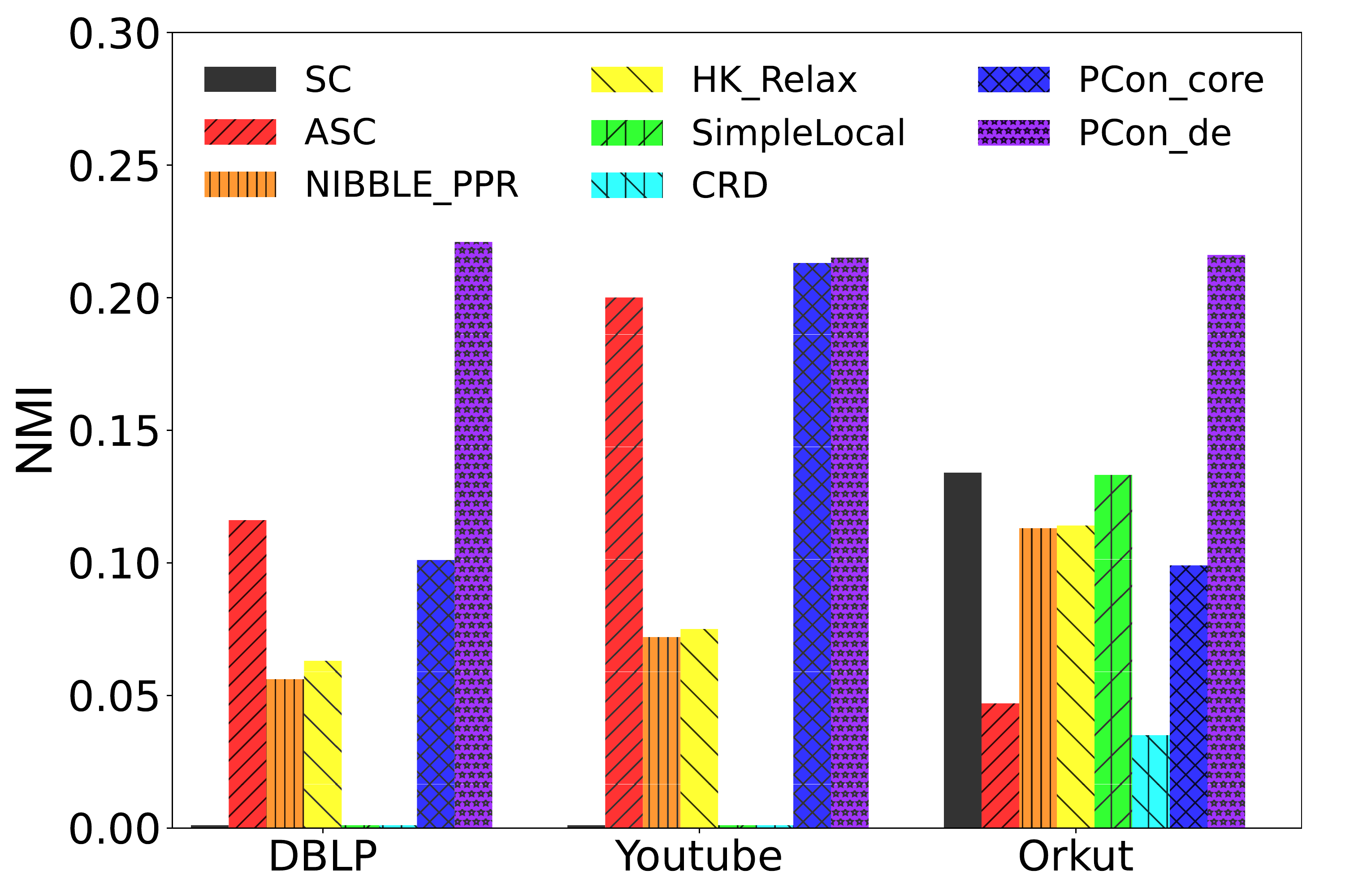}
	\vspace{-0.3cm}
	\caption{\textit{NMI} scores on real-world graphs with ground-truth clusters.}
	\label{fig:gr} \vspace{-0.3cm}
\end{figure}

\begin{figure}[t!]
	\centering
	\includegraphics[width=0.45\textwidth]{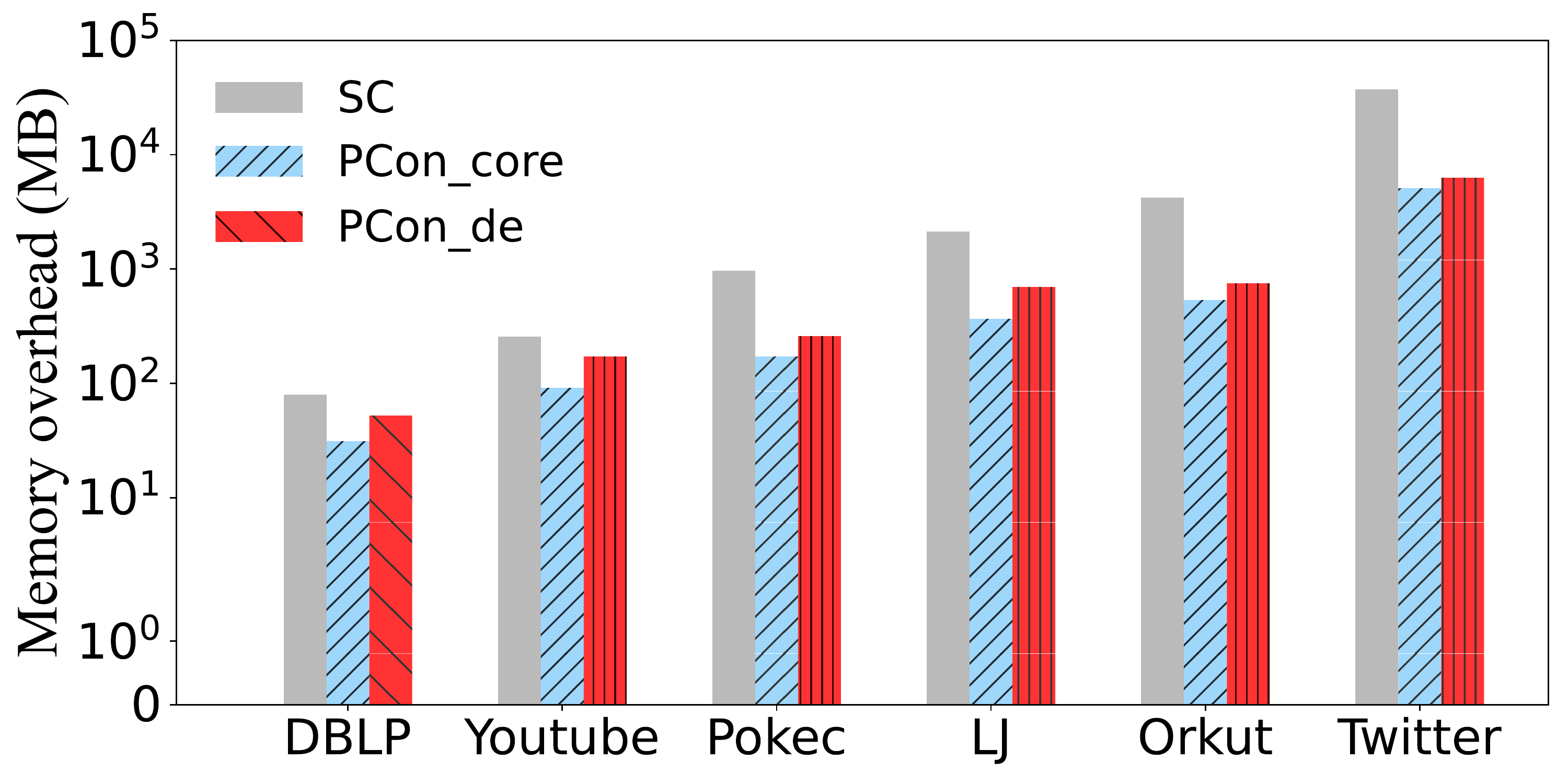}
	\vspace{-0.3cm}
	\caption{Memory overhead on real-world graphs (excluding the size of the graph itself).}
	\label{fig:memory} 
\end{figure}

\begin{table*}[t!]
	\caption{Runtime (in seconds) and conductance  on  real-world graphs. The best result in each metric is highlighted in bold.} \vspace{-0.2cm}
	\centering
	\scalebox{1}{
		\begin{tabular}{c|ccccccc}
			\toprule
			\multicolumn{1}{c|}	{Runtime/Conductance} &DBLP& Youtube&Pokec&LJ&Orkut&Twitter\\
			\midrule
			\multirow{6}{*} \textit{SC} &56/\textbf{0.009}&552/0.006&165/0.002&14492/0.001&837/0.007 &130570/ 0.002\\
			
			\multirow{6}{*} \textit{ASC} &21/0.482&73/0.554&730/0.454&1185/0.423&3089/0.415&67409/0.602\\
			
			\multirow{6}{*} \textit{NIBBLE\_PPR} &\textbf{7}/0.130 &\textbf{18}/0.110 &156/0.184& 784/0.021&3768/0.009&5822/0.177\\
			
			\multirow{6}{*} \textit{HK\_Relax} &31/0.127 &138/0.113&1271/0.011&2144/0.036&6016/0.008&99613/0.026\\
			
			\multirow{6}{*} \textit{SimpleLocal} &108/\textbf{0.009}&596/0.006&213/0.002& 14807/0.001&16267/0.013&283745/0.004\\			
			
			\multirow{6}{*} \textit{CRD} &47/0.184&235/0.183&1231/0.218
			&1279/0.194
			&7261/0.141
			&38060/0.336\\				
			\midrule
			
			\multirow{6}{*} \textit{PCon\_core} &9/0.106&27/0.404&\textbf{116}/0.327&\textbf{266}/0.088&\textbf{500}/0.223 &\textbf{5486}/0.419\\			
			
			\multirow{6}{*} \textit{PCon\_de} &11/0.027&39/\textbf{0.004}&141/\textbf{0.001}& 345/\textbf{0.000}&577/\textbf{0.006}&7655/\textbf{0.000}\\	
			\bottomrule	
	\end{tabular}}
	\label{table:metric}
\end{table*}

\begin{figure*}[t!]
	\centering
	\subfigure[Scalability testing on \textit{ER} synthetic graph]{
	\includegraphics[width=0.32\textwidth]{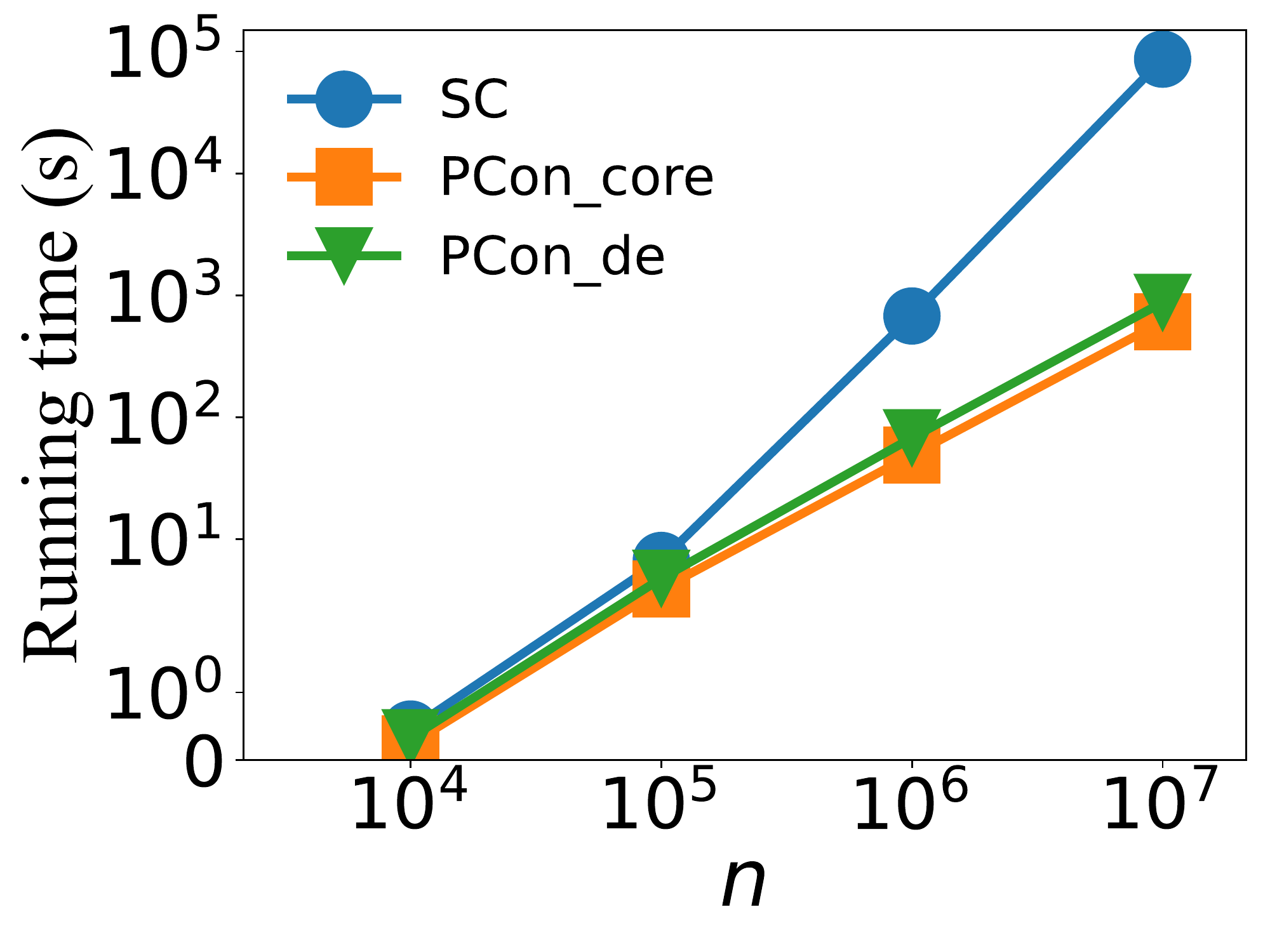}
	\label{fig:sca(a)}
}
	\subfigure[\textit{LFR} synthetic graphs varying $\mu$]{
		\includegraphics[width=0.32\textwidth]{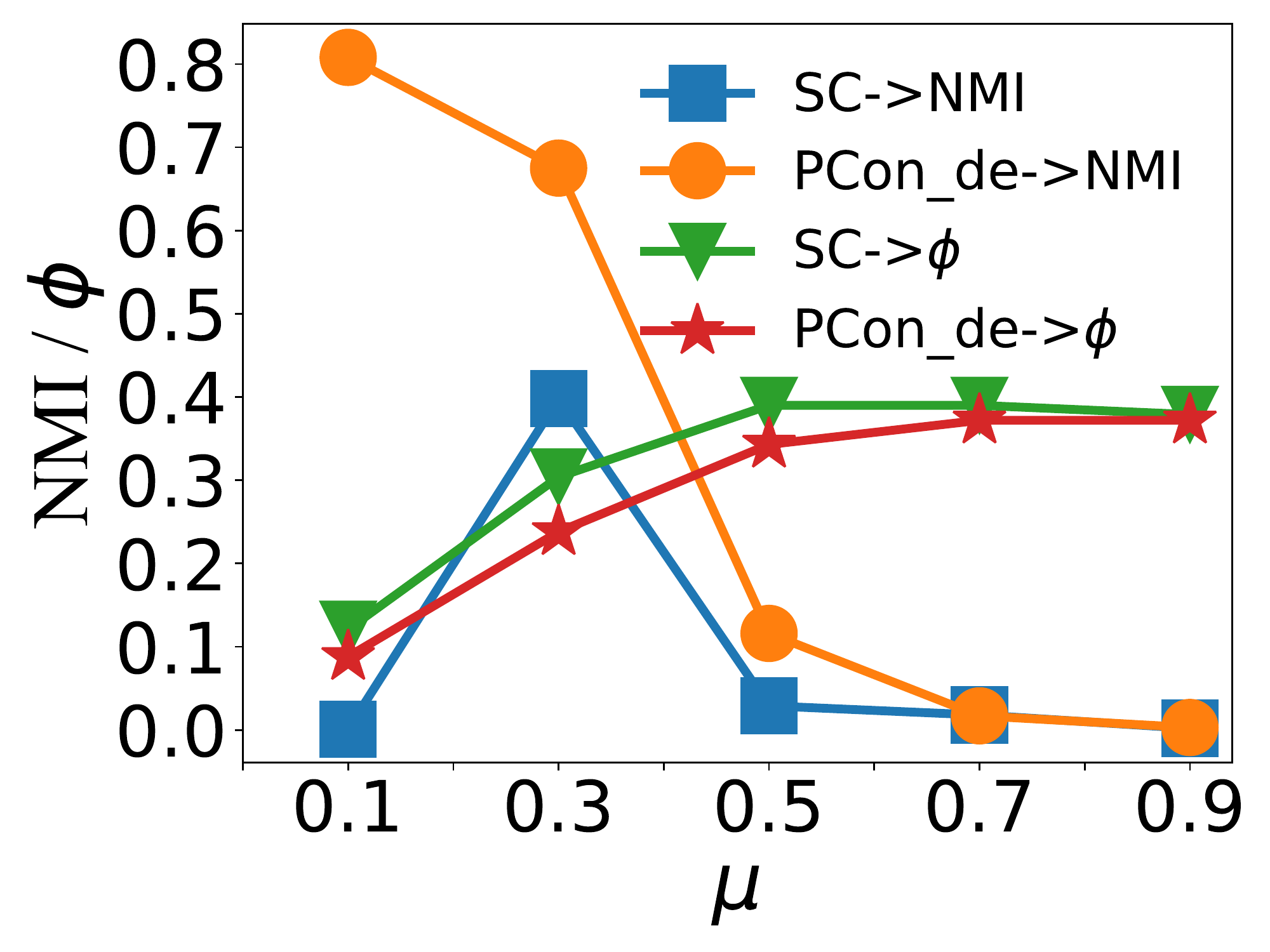}
		\label{fig:quality(a)}
	}
	\subfigure[\textit{LFR} synthetic graphs varying $n$]{
		\includegraphics[width=0.32\textwidth]{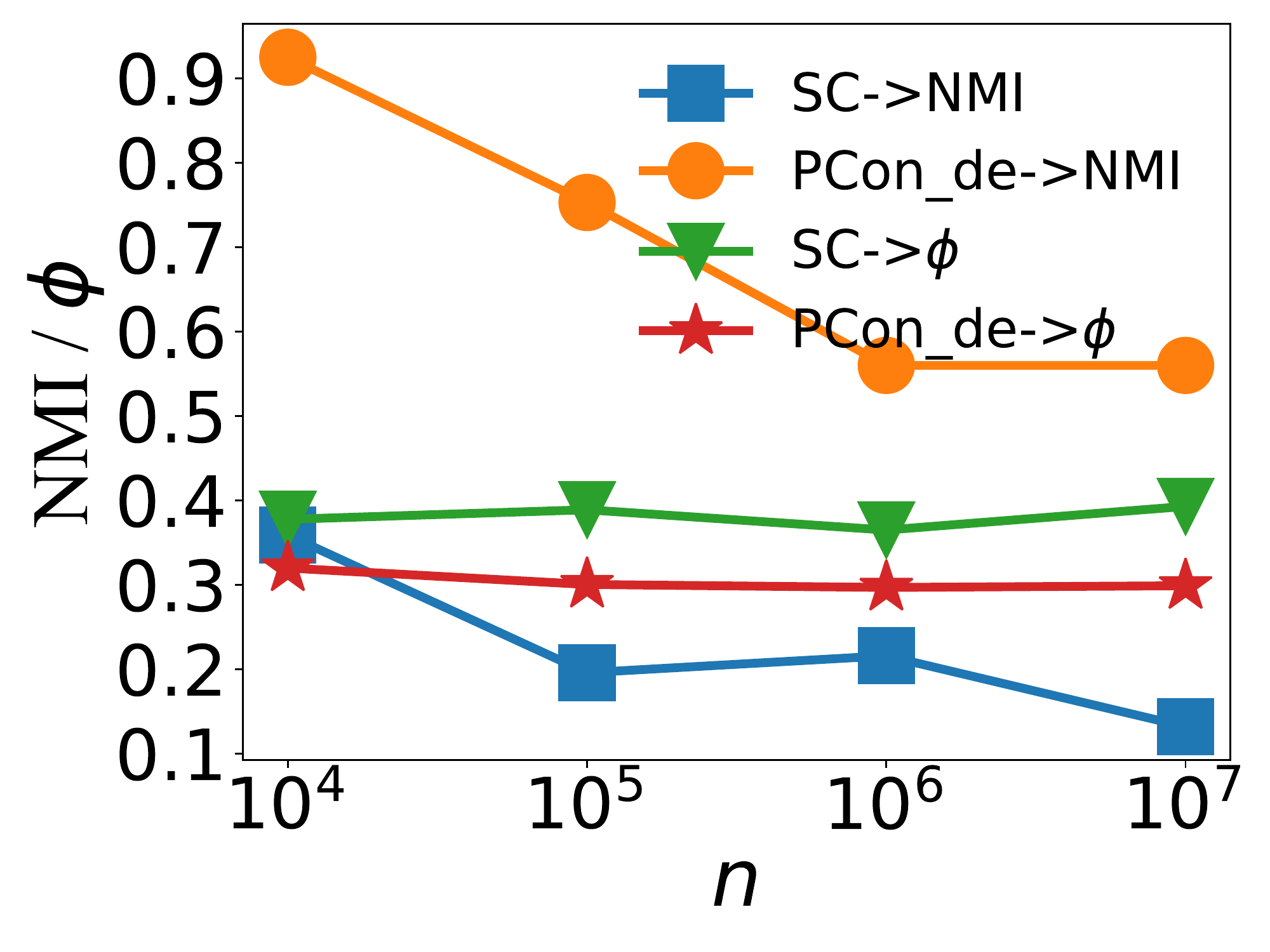}
		\label{fig:quality(b)}
	}
\vspace*{-0.3cm}
	\caption{Results on synthetic graphs.}
	\label{fig:quality}
\end{figure*}  

\subsection{Results on real-world graphs}  
Table \ref{table:metric} reports runtime and conductance for each method on real-world graphs. We have the following observations: (1) The proposed algorithms \textit{PCon\_core} and \textit{PCon\_de} are consistently faster than other methods on most graphs, especially for larger graphs. In particular, they achieve the speedups of 5, 14, 42, and 17 times over \textit{SC} on DBLP, Youtube, LJ, and Twitter, respectively.    For example, on Twitter with  more than a few billion edges, \textit{PCon\_core} and \textit{PCon\_de} take 5,486 seconds and 7,655 seconds to obtain the result, respectively, while \textit{SC} takes 130,570 seconds. (2) \textit{PCon\_core} is slightly fast than \textit{PCon\_de} but \textit{PCon\_core} has poor conductance value. (3) The clustering quality returned by \textit{PCon\_de} outperforms other methods on five of the six datasets. Besides, \textit{PCon\_de} always outperforms \textit{NIBBLE\_PPR} and \textit{HK\_Relax}. This is because  \textit{PCon\_de} can find clusters with near-liner approximation ratio, while \textit{Fiedler} vector-based spectral clustering algorithms (i.e., \textit{SC} and \textit{ASC}) have quadratic chegger bound and  diffusion-based local clustering algorithms (i.e., \textit{NIBBLE\_PPR} and \textit{HK\_Relax}) have no guarantee of clustering quality  (Table \ref{tab:alg}). Therefore, these results give clear evidences that the proposed algorithms can achieve significant speedup with high clustering quality compared with the baselines.

We  use the  normalized mutual information (\textit{NMI} for short) \cite{DBLP:journals/tit/CilibrasiV05} to measure how ``close" each detected cluster is  to the  ground-truth one. Note that for a cluster $C$,  the larger  the \textit{NMI}, the better the quality of the cluster $C$.  Figure \ref{fig:gr} shows \textit{NMI}  scores of different methods on real-world graphs with ground-truth clusters. As can be seen,  \textit{PCon\_de} consistently outperforms other methods. The \textit{NMI} scores of all methods other than \textit{PCon\_de} vary significantly depending on the dataset. For example, \textit{SC}, \textit{SimpleLocal}, and \textit{CRD}  are almost zero on DBLP and Youtube, but they have relatively good \textit{NMI} scores on Orkut. Therefore, these results  imply that \textit{PCon\_de}  approximates the ground-truth clusters more effectively than all other methods.

From Figure \ref{fig:memory}, we can see that memory overhead of \textit{PCon\_core} and \textit{PCon\_de} are consistently less than \textit{SC} on all datasets. In particular, they save 1.4 $\sim$ 7.8 times  memory overhead compared to \textit{SC}.  For example, on Orkut graph, \textit{PCon\_core} and \textit{PCon\_de} consume about 535.88MB  and 746.26MB to obtain the result, respectively, while \textit{SC} consumes 4209.25MB. This is because \textit{PCon\_core} and \textit{PCon\_de} only need linear space complexity to calculate the score of each vertex. However,  \textit{SC} adopts the eigenvector of the matrix to compute the score of each vertex, thus it requires squared space complexity in the worst case. These results demonstrate that \textit{PCon\_core} and \textit{PCon\_de} can identify clusters with low memory.

\subsection{Results on synthetic graphs}
We further use synthetic graphs to evaluate the scalability and  effectiveness of our proposed solutions. In particular, we only present the scalability for  \textit{ER} \cite{erdos1960evolution} with varying the number of vertices in Figure \ref{fig:sca(a)}, but all other  synthetic graphs have similar trends. As can be seen, \textit{PCon\_core} and \textit{PCon\_de} scale near-linear with respect to the size of the graphs. However, \textit{SC} has poor scalability because its time cost fluctuates greatly as the graph size increases.  These results indicate that our proposed algorithms can handle massive graphs while \textit{SC} cannot.

The \textit{LFR} model \cite{lancichinetti2009detecting} is a widely used benchmark with ground-truth clusters. In Figure \ref{fig:quality(a)}, we generate five synthetic graphs composed of 1,000,000 vertices by varying the mixing ratio $\mu$ from 0.1 to 0.9. A larger $\mu$ implies that the number of edges crossing clusters increases, resulting in that being more difficult to detect intrinsic clusters. As can be seen,  \textit{PCon\_de} consistently outperforms \textit{SC} in terms of \textit{NMI} and $\phi$. Meanwhile, the quality of \textit{PCon\_de} decreases as $\mu$ increases, but it is always better than \textit{SC}. Furthermore, we also generate four graphs with varying the number of vertices when fixing $\mu=0.4$. Figure \ref{fig:quality(b)} shows a similar trend as Figure \ref{fig:quality(a)}. As a consequence, these results  imply that \textit{PCon\_de}  approximates the ground-truth clusters more effectively than \textit{SC}.

\section{Conclusion}
In this paper, we devise a \textit{peeling}-based  computing framework \textit{PCon} for conductance-based graph clustering. We observe that most state-of-the-art algorithms can be reduced to \textit{PCon}. Inspired by our framework, we first propose an efficient heuristic algorithm  \textit{PCon\_core}, which adopts the degeneracy ordering to detect clusters. To improve the accuracy, we further propose a powerful \emph{PCon\_de} with near-constant approximation ratio, which achieve an important theoretical improvement over existing approaches such as \textit{Fiedler} vector-based spectral clustering. Finally, extensive experiments on eleven datasets with six competitors demonstrate that the proposed algorithms can achieve 5$\sim$42 times speedup  with a high accuracy and 1.4$\sim$7.8 times less memory than  the state-of-the-art solutions.

\section{Acknowledgments}
This work is supported by (i) National Key R$\&$D Program of China 2021YFB3301300, (ii) NSFC Grants U2241211, 62072034,  U1809206, (iii) Fundamental Research Funds for the Central Universities under Grant SWU-KQ22028, (iv) University Innovation Research Group of Chongqing CXQT21005, (v) Industry-University-Research Innovation Fund for Chinese Universities  2021ALA03016, and (vi) CCF-Huawei Populus Grove Fund. Rong-Hua Li is the corresponding author of this paper.

\bibliography{aaai23}

\begin{table*}[t!]
	\caption{Conductance of diffusion-based local clustering with varying parameter $\epsilon$.} \vspace{-0.2cm}
	\centering
	\scalebox{1}{
		\begin{tabular}{c|ccccccc}
			\toprule
			\multicolumn{1}{c|}	{Dataset} &Methods &$\frac{10}{m}$& $\frac{100}{m}$&$\frac{1000}{m}$&$\frac{10000}{m}$&$\frac{100000}{m}$&$\frac{1000000}{m}$\\
			\midrule
			\multirow{2}{*}{DBLP}  & \textit{NIBBLE\_PPR} &0.237&0.331&0.345&0.394&0.833 &1.000\\
			&\textit{HK\_Relax}	&0.138&0.275&0.324&0.394&0.240&0.410\\
			\midrule
			\multirow{2}{*}{Youtube}  &\textit{NIBBLE\_PPR}&0.295&0.281&0.478&0.704&0.888 &0.777\\
			&\textit{HK\_Relax}	&0.121&0.332&0.197&0.360&0.461 &0.777\\
			\midrule
			\multirow{2}{*}{Pokec}  &\textit{NIBBLE\_PPR}&0.587&0.741&0.667&0.761&0.796 &1.000\\
			&\textit{HK\_Relax}&0.051&0.243&0.437&0.691&0.615 &0.886\\
			\midrule
			\multirow{2}{*}{LJ}  & \textit{NIBBLE\_PPR} &0.244&0.439&0.486&0.580&0.685 &0.783\\
			&\textit{HK\_Relax}	&0.038&0.1998&0.170&0.387&0.294 &0.585\\
			\midrule
			\multirow{2}{*}{Orkut}  &\textit{NIBBLE\_PPR}&0.654&0.886&0.767&0.924&0.913 &0.974\\
			&\textit{HK\_Relax}	&0.008&0.036&0.409&0.690&0.649 &0.705\\
			\midrule
			\multirow{2}{*}{Twitter}  &\textit{NIBBLE\_PPR}&0.819&0.906&0.912&0.931&0.947&0.952\\
			&\textit{HK\_Relax}	&0.054&0.102&0.205&0.301&0.403 &0.506\\
			\bottomrule	
	\end{tabular}}
	\label{table:metric_add}
\end{table*}

\begin{figure*}[t!]
	\centering
	\subfigure[\textit{LFR} synthetic graph]{
		\includegraphics[width=0.22\textwidth]{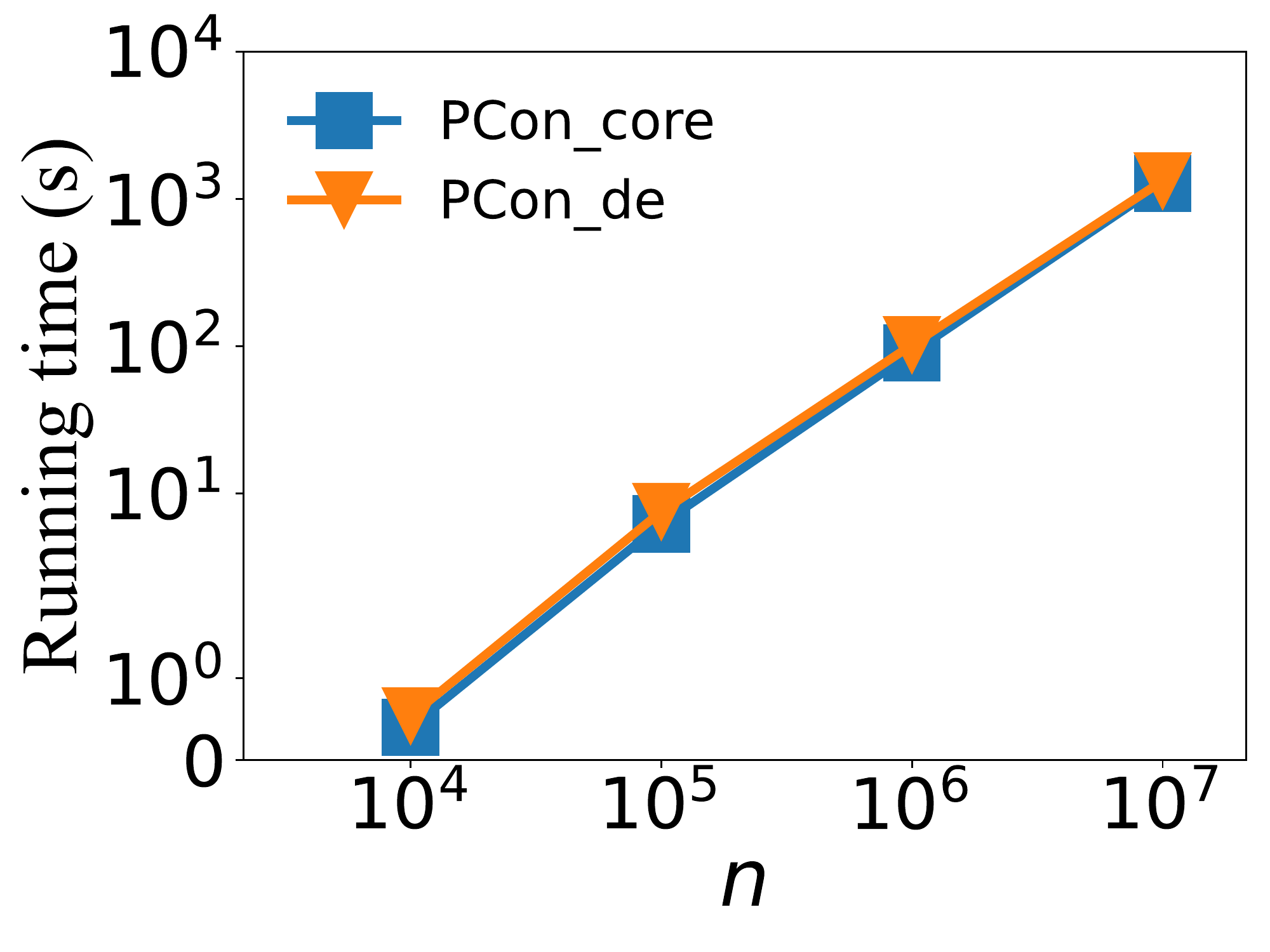}
		\label{fig:sca1(a)}
	}
	\subfigure[\textit{WS} synthetic graph]{
		\includegraphics[width=0.22\textwidth]{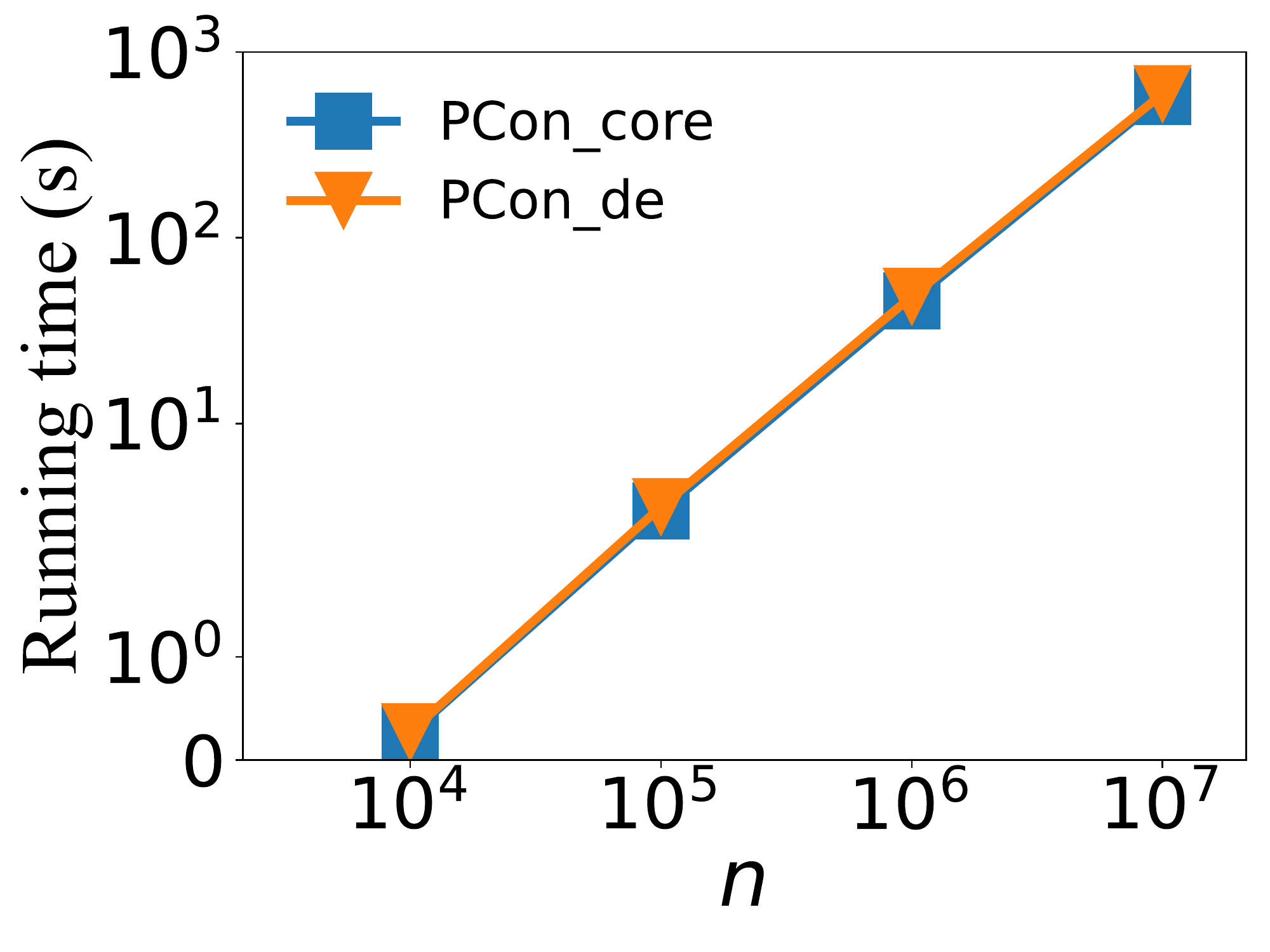}
		\label{fig:sca1(b)}
	}
	\subfigure[\textit{PLC} synthetic graph]{
		\includegraphics[width=0.22\textwidth]{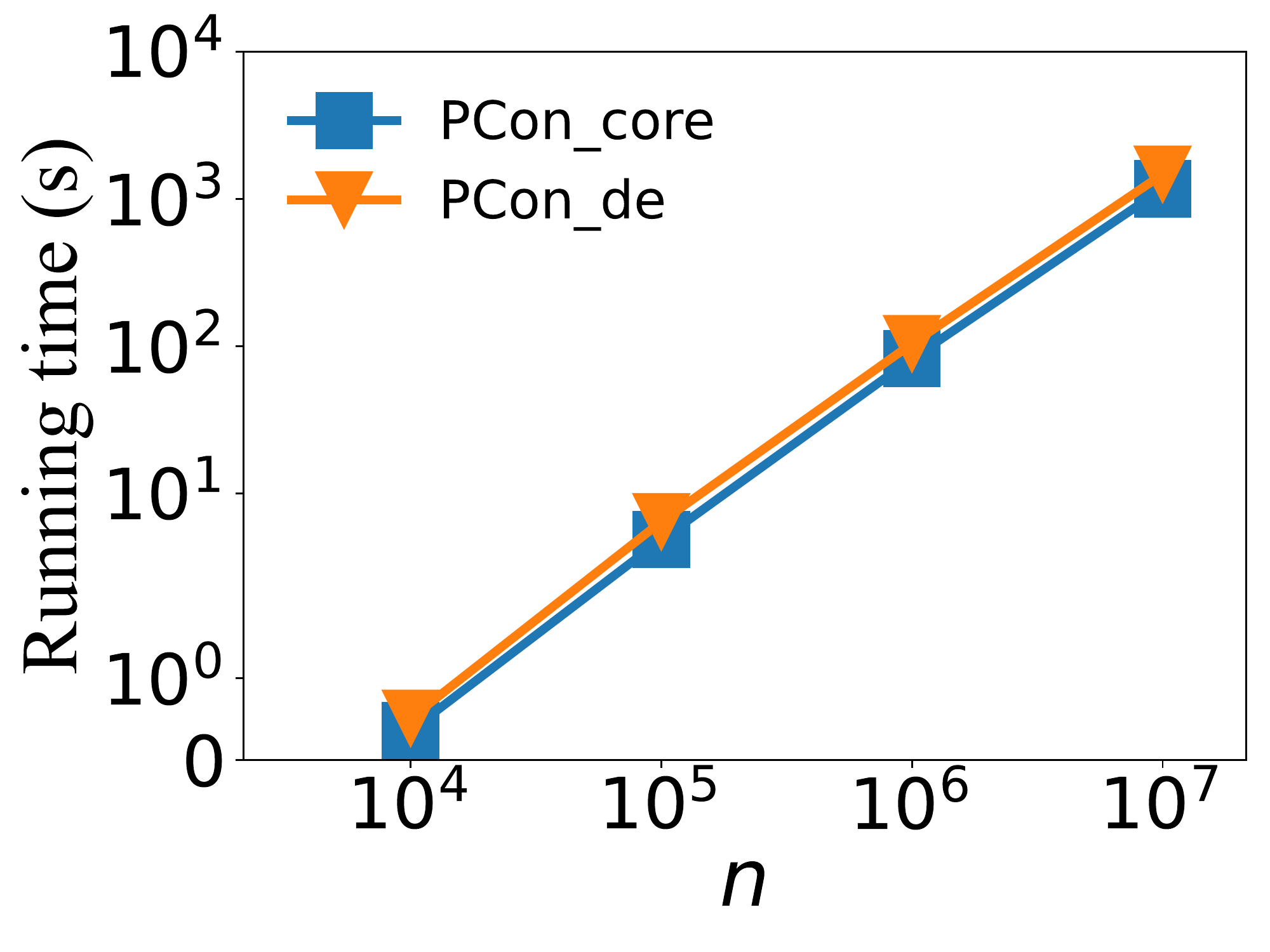}
		\label{fig:sca1(c)}
	}
	\subfigure[\textit{BA} synthetic graph]{
		\includegraphics[width=0.22\textwidth]{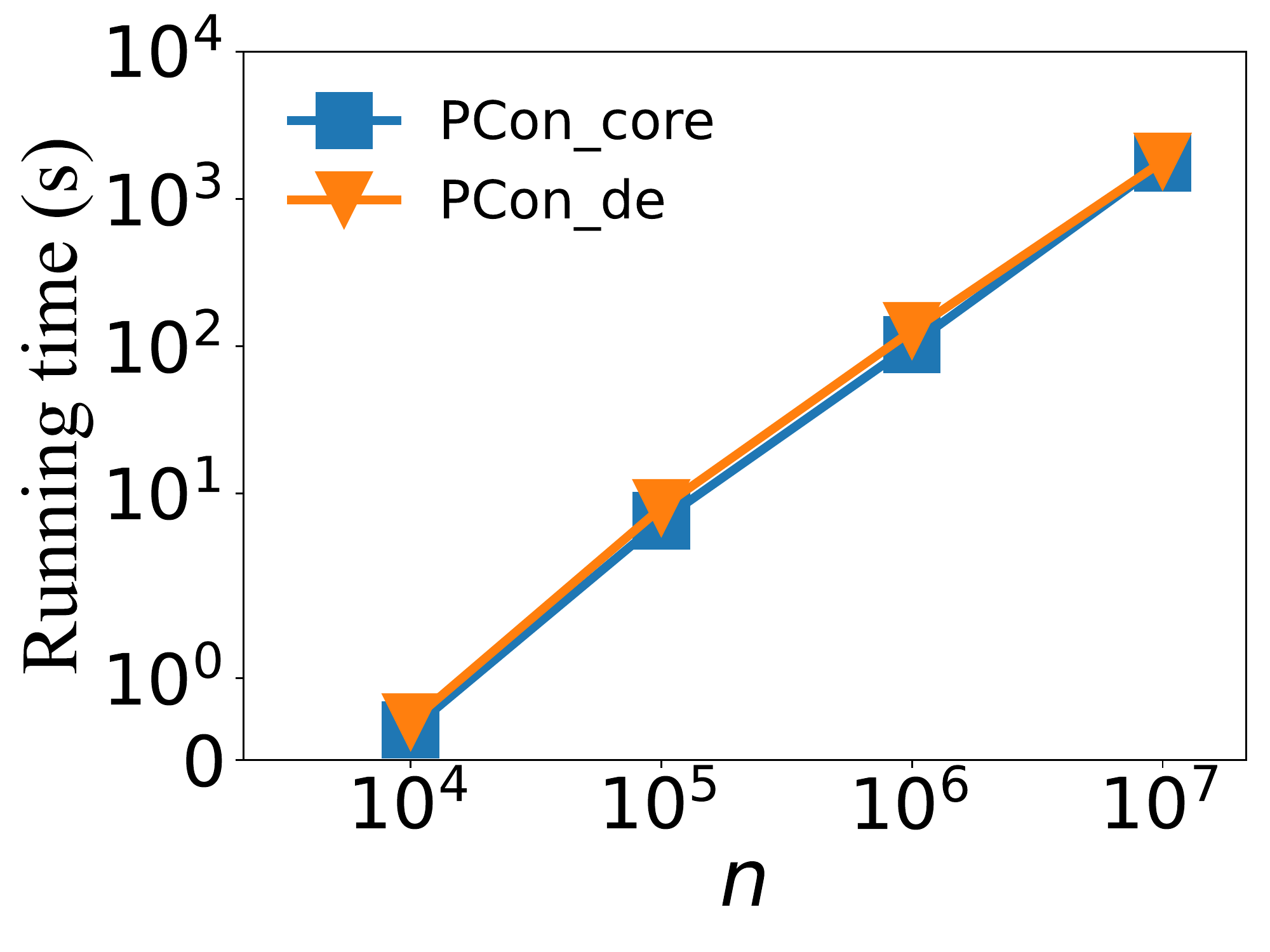}
		\label{fig:sca1(d)}
	}
	\vspace*{-0.3cm}
	\caption{Scalability testing on synthetic graphs.}
	\label{fig:scal}
\end{figure*}

\subsection{Additional results}

Table \ref{table:metric_add} reports the conductance quality   of diffusion-based local clustering with varying parameter $\epsilon$. Specifically, we vary  $\epsilon$ in $\{\frac{10}{m},\frac{100}{m},\frac{1000}{m},\frac{10000}{m},\frac{100000}{m},\frac{1000000}{m}\}$, in which  $m$ is the number of edges of the corresponding graph. As can be seen, \textit{PCon\_core} outperforms diffusion-based local clustering in most parameter settings (see Table \ref{table:metric}  and Table \ref{table:metric_add}). Furthermore, since the parameter $\epsilon$ controls how much of the graph is explored by their respective method, the conductance quality gets  worse as $\epsilon$ increases in most cases. Thus,
these experimental results imply that \textit{PCon\_core} is faster in time than \textit{Fiedler}
vector-based spectral clustering and better in conductance quality than diffusion-based local clustering.

We use four synthetic graphs \textit{LFR} \cite{lancichinetti2009detecting}, \textit{WS} \cite{watts1998collective}, \textit{PLC} \cite{holme2002growing}, and \textit{BA} \cite{barabasi1999emergence} to evaluate the scalability of our proposed solutions. Figure \ref{fig:scal} shows the results. As can be seen, \textit{PCon\_core} and \textit{PCon\_de} scale near-linear with respect to the size of the graphs. Thus, these results indicate that our proposed algorithms can handle massive graphs.

\end{document}